\documentclass[runningheads]{llncs}
\usepackage{graphicx}
\usepackage{cite}
\usepackage{lineno,hyperref}
\usepackage{setspace}
\usepackage{algpseudocode}
\usepackage{algorithm}
\usepackage{float}
\usepackage{amsmath}
\usepackage[export]{adjustbox}
\usepackage{subfig}

\usepackage{amsthm}
\usepackage{amssymb}
\usepackage{varwidth,comment}
\usepackage{graphicx}
\usepackage{multirow}
\usepackage{pgfplots}
\usetikzlibrary{pgfplots.groupplots}
\usepackage{tikz,array,calc}
\usetikzlibrary{decorations.pathreplacing}
\usepackage{caption}
\begin{document}
\title{High-utility itemset mining for subadditive monotone utility functions
	}
%
%

\author{Siddharth Dawar, Debajyoti Bera, Vikram Goyal}
\institute{Indraprastha Institute of Information Technology, Delhi, India
\email{\{siddharthd,dbera,vikram\}@iiitd.ac.in}}

%

%
\maketitle              
\begin{abstract} 
High-utility Itemset Mining (HUIM) finds itemsets from a transaction database with utility no less than a user-defined threshold where the utility of an itemset is defined as the sum of the item-wise utilities. In this paper, we generalize this notion to utility functions that need not be a simple sum of individual utilities. In particular, we study generalized utility functions that are subadditive and monotone (SM). We also describe a novel function that allows us to include external information in the form of a relationship graph for computing utility. Next, we focus on algorithms for HUIM problems with SM utility functions. We note that the existing HUIM algorithms use upper-bounds like ``Transaction Weighted Utility'' and ``Exact-Utility, Remaining-Utility'' for efficient search-space exploration. We derive analogous and tighter upper-bounds for SM utility functions. We design a novel inverted-list data structure called SMI-list and a new algorithm called SM-Miner to mine HUIs for SM functions. We explain how existing tree-based and projection-based HUIM algorithms can be adapted using these bounds. We experimentally compare adaptations of some of the latest HUIM algorithms and point out some caveats that should be kept in mind while handling utility functions that allow integration of domain knowledge with a transaction database.   

\end{abstract}
\section{Introduction}
\label{intro}
High-utility itemset mining (HUIM) finds itemsets from a transaction database with utility no less than a user-defined threshold \cite{two_phase}. Every item in a transaction is associated with a quantity and global importance, say profit. The utility of an item is defined as the product of its quantity and profit. The utility of {\em an itemset in a transaction} is the {\em sum} of the utilities of individual items. HUIM can be used to identify sets of items that generated a large profit over many transactions.

HUIM has also been applied to extract co-expression patterns from gene-expression data \cite{biomedical_appln}. With the penetration of data mining techniques in hitherto untrodden areas, we envisage scenarios where ``utility'' of an itemset (in a transaction) is not a simple addition of the utilities of the individual items. We investigate the problem of high-utility itemset mining for utility functions that are subadditive and monotone (HUIM-SM). Our results can be seen as a generalization of HUIM since summation (of utilities) is such a function. While monotonicity appears to be a natural requirement, subadditivity stipulates the utility of a combination of items cannot be more than the combination of individual utilities. Subadditivity can be observed at several places in the retail domain, e.g., customers frequently purchase a set of products at less price compared to the sum of the individual prices. Retail stores often offer discounts \cite{bansal2015efficient} to customers on purchasing a hamper of products. 

Generalization of the utility function can lead to interesting applications of itemset mining techniques firstly by allowing novel mapping of quantities to utility and furthermore, by allowing integration of additional information (say, on items) for computing utility. For an example of the latter, consider a social network where the interests of every user are known. Suppose one transaction is created for every distinct interest, e.g., comprising of members who like ``chess'', and suppose the utility of a set of members is some clustering score (e.g., network modularity \cite{modularity} of that group. Then, HUIM-SM will reveal groups of users that have a high overall modularity score across all interests. Subadditive monotone set functions have been observed in many application domains to identify important groups of entities. Chen et al. \cite{chen} extended the concept of closeness centrality \cite{cc} to find a set of k nodes that have the largest closeness centrality as a whole. Their paper proved that closeness centrality for a set of nodes is subadditive \footnote{A non-negative submodular function is also subadditive but the reverse need not be true.} and monotone. Yan et al. \cite{yan} investigated the problem of group-level influence maximization with budget constraints in a social network. They proposed an subadditive monotone function to capture the influence spread scope for a group of users, and an algorithm that returns a group close to optimal set of users that maximize the influence spread with a limited cost budget. Tschiatschek et al. \cite{Tschiatschek} introduced a novel set of subadditive monotone utility functions over a sequence of items that captures the ordered preferences among items over arbitrarily long ranges. They proposed a greedy algorithm for selecting sequence of items under sequence length constraints, and showed an application of their algorithm on a movie recommendation dataset for the task of recommending a sequence of movies to a user. Perozzi et al. \cite{Perozzi} proposed a subadditive monotone objective function to select a small number of communities from a large number of extracted communities from an input graph that represents the input graph. Parambath et al. \cite{parambath} proposed a unified framework and an algorithm for the problem of group recommendation where a fixed number of items can be recommended to a group of users. Their paper framed the problem as choosing a subgraph with the largest group consensus score in a completely connected graph defined over the item affinity matrix. They proposed a subadditive monotone group consensus objective function for their problem setting.

With general utility functions in mind, we explored how they can be incorporated within existing HUIM algorithms. We ask the question if, when, and how, can we design a high-utility itemset mining algorithm for any arbitrary subadditive monotone function? Will the existing bounds used for search-space exploration still work? Our specific contributions are summarized below.
\begin{enumerate}
	\item We define the problem of high-utility itemset mining for subadditive monotone utility functions \cite{yan,Tschiatschek,Perozzi}. We investigate whether the existing tree-based, projection-based and list-based algorithmic frameworks for HUIM are sufficient to design algorithms for HUIM with SM functions (HUIM-SM). These frameworks are fuelled by ``upper bounds'' like ``Transaction-weighted utility ($TWU$)" and "Exact-utility, Remaining-utility ($EU_-RU$)". We derive new upper bounds (TSMWU and CU) that are better than TWU and $EU_-RU$ and use them to show how to adapt existing tree-based and projection-based algorithms to HUIM-SM.
	\item For the list-based paradigm, we observe that the existing list-based data structures in HUIM are unsuitable for HUIM-SM. So, we design a novel inverted-list data structure called SMI-list with a lightweight construction method and an algorithm called SM-Miner to mine high-utility itemsets for subadditive monotone utility functions.
	\item While traditional HUIM defines utility solely on the transaction database, a generalized utility can depend upon external information about the items and their quantities. We design a function $ucov$ that is also subadditive and monotone and captures the relationship among the items in the form a relationship graph.
	\item We conduct experiments with several HUIM-SM algorithms on real dense and sparse datasets. Our results demonstrate that the computation of utility can be a big factor in the relative performance of algorithms for complex functions like $ucov$. We highlight that subadditive monotone functions like $ucov$ can be designed to find high-influential groups of active users that can be attractive for applications like viral marketing. 
\end{enumerate}

\section{Related work}\label{sec:Related Work} 
Several algorithms have been proposed in the literature for HUIM. Broadly these algorithms vary in terms of the number of stages they run for (one phase and two-phase), their database representation (tree, utility list, and projected database) and other data structures and heuristics used to prune non high-utility itemsets effectively.

Liu et al. \cite{two_phase} developed a two-phase algorithm for mining high-utility itemsets from a transaction database. The two-phase algorithm overestimates the utility of all supersets of an itemset $X$ by computing transaction-weighted utility(TWU) \cite{twdc} as the utility measure in HUIM is not anti-monotonic. Itemset $X$ and its supersets can't be high-utility itemsets if transaction-weighted utility of $X$ is less than the minimum utility threshold. The two-phase algorithm generates a set of candidate high-utility itemsets in a level-wise manner like Apriori \cite{apriori}. A database scan is performed in the next phase to filter out the high-utility itemsets from the candidates; however, a lot of database scans are performed during the candidate generation phase. 

Ahmed et al. \cite{ihup_tree} proposed the first tree data structure called IHUP-tree, and a recursive tree-based algorithm for mining high-utility itemsets from a transaction database using only two database scans. Every node in the IHUP-tree stores the item name, a frequency count, and a TWU value. Every transaction in the transaction database is inserted to form a IHUP-tree. A path from the root to the leaf node represents a set of transactions. The IHUP mining algorithm generates a set of candidate high-utility itemsets by recursively creating local tree structures from the global IHUP tree for the itemsets with transaction-weighted utility \cite{twdc} no less than the minimum utility threshold. Another database scan is performed to find the exact high-utility itemsets by computing the exact utility of the generated candidates in the previous phase. Mining algorithm based on IHUP-tree generates a lot of candidates in the first phase, and the number of candidates influences the performance of tree-based algorithms. Tseng et al. \cite{upgrowth} proposed few strategies to reduce the estimated utility value of itemsets during the candidate generation phase, and a tree structure called UP-tree that stores a node utility value along with each node instead of TWU stored by IHUP-tree. An algorithm called UP-Growth was proposed similar to IHUP mining algorithm to generate candidates in the first phase and verify them later. Tseng et al. \cite{up_tree} proposed another algorithm called UP-Growth+ that produced fewer candidates compared to UP-Growth by storing better utility estimate with each node of the UP-tree. Dawar et al. \cite{uphist} proposed another tree structure called UP-Hist that stores a histogram of item-quantities with each node to further reduce the number of candidates generated in the first phase. 

Tree-based algorithms spend a lot of time verifying candidates to generate the complete set of high-utility itemsets. Liu et al. \cite{hui_miner} proposed the first algorithm called HUI-Miner for mining high-utility itemsets without generating any candidates. A list-based data structure called utility-list was proposed that stores a Transaction identifier (TID), Exact-utility (EU), and Remaining-utility (RU) for each itemset. Initially, a database scan is performed to remove items with TWU less than the minimum utility threshold. Another database scan is performed to construct the utility-list for every item remaining in the database. The utility-list for a $\lbrace k \rbrace$-itemset is constructed from the utility-lists of two $\lbrace k-1 \rbrace$-itemsets. HUI-Miner can compute the exact-utility of an itemset by constructing its utility-list, and avoids the bottleneck phase of candidate verification. The operation of joining utility-list of two itemsets to create the utility-list for longer itemset is a costly operation. Fournier-Viger et al. \cite{fhm} proposed an ``Efficient Utility Co-occurrence Structure'' (EUCS) strategy to reduce the number of join operations by keeping the TWU information for every pair of items and an algorithm called FHM that makes use of the EUCS strategy. Dawar et al. \cite{hybrid} proposed a hybrid algorithm called UFH based on integration of UP-Growth+ and FHM that generates high-utility itemsets in a single phase only like list-based algorithms and reduces the number of join operations performed during the construction of utility-lists. 

Zida et al. \cite{efim} proposed a fast and memory-efficient algorithm called EFIM that stores the transactions as a projected database recursively during the mining process. EFIM was shown to be about two to three orders of magnitude faster than UP-Growth+, FHM, and several other algorithms on dense datasets. Liu et al. \cite{d2hup} proposed a data structure called Chain of Accurate Utility Lists (CAUL) and an algorithm called D2HUP for efficiently mining high-utility itemsets from sparse datasets. Krishnamoorthy \cite{hminer} proposed another data structure based on utility-list and hyperlink structure used by D2HUP and an algorithm called HMINER to extract high-utility itemsets. 

However, all of the above techniques work for a specific utility function and do not generalize to subadditive monotone utility functions.

Guns et al. \cite{cp_guns_2011} proposed a declarative constraint programming approach to model and solve tasks like frequent itemset mining \cite{apriori}, discriminative \cite{discriminative}, closed itemset mining \cite{closed} etc. The itemset mining problem is expressed in a high-level language, and a solver is used to perform the mining task. Constraint programming approach separates the model from the solver. It was observed that the constraint programming approach performs worse than specialized algorithms for frequent itemset mining due to the use of alternative data structure and other overhead in the constraint programming system. Silva et al. \cite{cpm_survey} proposed a framework for constraint pattern mining that allowed to organize and analyze different algorithms based on the properties of constraints like anti-monotonicity, monotonicity, succinctness, prefix-monotonicity, and mixed monotonicity. Guns et al. \cite{miningzinc} introduced a declarative framework named MiningZinc for constraint-based data mining. To the best of our knowledge, MiningZinc is the first framework that can express the high-utility itemset mining problem. MiningZinc uses a cover function that returns the set of transaction identifiers containing an itemset and the actual data to compute its utility. MiningZinc expresses the itemset mining problem over integer and set variables with no explicit data structures. Coussat et al. \cite{high_utility_tensors} defined the problem of high-utility itemset mining in uncertain tensors, and showed that an algorithm called multidupehack \cite{multidupehack_2014} could be deployed to mine itemsets. The authors studied a generalized version of high-utility itemset mining problem where the utilities can be positive and negative and are not restricted to matrices. The utility of an itemset is defined as the sum of utility of individual items by MiningZinc and multidupehack. Computing the utility of an itemset as the sum of utility of items does not necessarily hold in our generalization of high-utility itemset mining for subadditive monotone function. 

The research community has contributed significantly to improve the efficiency of high-utility itemset mining algorithms in the last few years. The focus of this paper is to identify the changes required in the existing high-utility itemset mining algorithmic frameworks to mine itemsets for any subadditive monotone utility function.   

\section{Background and Problem Statement}\label{sec:Background}
Consider the usual setting of HUIM on a transaction database $D$ with $n$ transactions $\{T_1, T_2, \ldots, T_n\}$ over items $I=\lbrace i_1,i_2,...,i_m \rbrace$. Each transaction can be thought of as a subset of $I$ along with a positive quantity (or weight) associated with every item $i \in T$; we will use $q(i,T)$ to denote this weight and we will express $T$ as $\{ (i_{j_1} ~:~n_{j_1}), (i_{j_2}~:~n_{j_2}), \ldots (i_{j_k}~:~n_{j_k}) \}$ where $i_j$ denotes an item in $T$ and $n_j=q(i_j,T) \ge 0$ is its weight in $T$. By an ``itemset $X$ in $T$'' we will mean a set of some items appearing in $T$ and associated with the same weight as that in $T$. We will use  ``weighted itemset'' to refer to such sets with quantities.

For this work, we assume that there is some ``subadditive monotone utility'' function $f(\cdot)$ that returns a numerical value for any weighted itemset. For example,  $f\left(\{ (x_1:n_1), (x_2:n_2), \ldots, (x_k:n_k)\}\right)$ could be its cardinality $k$, the addition function over the quantities $\sum_j n_j$ or some complex function depending upon the items and their quantities.

Now, let $T$ be some transaction; for any subset $X = \{x_1, \ldots\}$ of items from $T$, we can define the utility of $X$ in $T$, denoted by $u(X,T)$, as the value of $f(\cdot)$ on the weighted itemset $\left\{ \big(x_1, q(x_1,T)\big), \ldots \right\}$.
The utility of $X$ in the database is defined as in HUIM: $u(X) = \sum_{\substack{X \subseteq T \\ T \in D}} u(X,T)$.\\

\noindent{\bf HUIM-SM problem:}
	An itemset $X$ is called a high utility itemset if $u(X)$ is no less than a given minimum user-defined threshold denoted by $\theta$. Given a transaction database $D$, a subadditive monotone utility function $f(\cdot)$ and a minimum user-defined threshold $\theta$, the aim is to find all high-utility itemsets.

Current HUIM algorithms work by recursively growing an itemset (called as the prefix) by appending an item from the set of items (I) to the prefix and then determining whether the newly formed prefix (called $\beta$) is high-utility or not. For this, a projected database for $\beta$ is constructed that contains all transactions with itemset $\beta$, and the algorithm recursively extends $\beta$. Since the search space of prefixes is exponential in the number of distinct items present in the database, the computational challenge is to efficiently determine whether there could be any high-utility itemset containing $\beta$. For this purpose, current HUIM algorithms use upper bound functions of $u(\beta)$ to decide for further exploration.

Even if the current HUIM algorithms are adapted for HUIM-SM, it is unclear if the current bounds will be correct for a different, general, utility function. Since a correct, and preferably tight, bound is crucial for efficient pruning of search space and also depends on the data structures used, careful attention must be given to come up with an appropriate bound function. We undertake this task for certain types of utility functions in the next few sections.

\section{Subadditive and monotone utility functions}\label{sec:our} 
Designing a generic yet efficient HUIM algorithm is tricky since it has to work for all utility functions. We leave that question for future, but in this work, we show how to design HUIM algorithms for any {\em subadditive and monotone} utility function (HUIM-SM) --- such functions include the $Sum$ utility function (defined below) that is used to define HUIM. For the definitions below, $I$ is the universe of items and $T$ is some transaction, essentially a weighted itemset.

\begin{definition}[subadditive and monotone function]
    Consider functions that map weighted itemsets over $I$ to positive real numbers. 
	Such a function $f(\cdot)$ is subadditive if $\forall ~X, Y \subseteq T$, $f(X\cup Y) \leq f(X)+f(Y)$. $f(\cdot)$ is monotone if $\forall X \subseteq Y \subseteq T$, $f(X) \leq f(Y)$. A utility function $u(\cdot,T)$ is subadditive and monotone if satisfies the above for {\em every transaction}.
\end{definition}
It should be noted that we require $u(X,T)$ to be subadditive and monotone {\em for any $T$} but no such property may hold for $u(X)$ i.e. the utility of an itemset $X$ in the database.

We now give a few examples of subadditive and monotone utility functions. First is $Sum(\{x_1:n_1, \ldots, x_k:n_k\}) = \sum_j n_j$, essentially returning the total quantity of items in an itemset. One should note that this is the utility function used in HUIM. Two other examples are $\log\left( \prod_j n_j \right)$ and $\sqrt{\sum_j n_j}$. It is known that the last function is subadditive and monotone \cite{penalty} and these properties are easy to verify for the first two functions as well.

\subsection{Graph-based utility function}
The examples of utility functions on weighted itemsets stated earlier depend on the quantities of included items but did not depend upon a particular grouping of items. Now we describe a utility function $ucov$ that not only depends upon the items but also allows us to incorporate any additional information on the relationship between the items into the utility function. Referring to the applications discussed in Section \ref{intro}, such knowledge may be helpful in mining retail logs, and finding active communities in a social network. Having seen the popularity of subadditive and monotone functions in several applications as described in Section \ref{intro}, we ensured that our utility function belongs to the same flock as well.

\begin{table}
	
	\begin{center}
		\caption{$Transaction \, database \, for \, coverage \, utility (ucov)$ \label{fig:Example_2}}
		\scalebox{0.75}
		{
		\begin{tabular}{  l  p{5.7cm} p{1cm} l p{2.2cm} l }
				\hline 
				\bfseries{TID} & \bfseries{Transaction} & \bfseries{$TU$} & \bfseries{$TSMU$} & \bfseries{$EU(\lbrace AC \rbrace,T)$ + $RU(\lbrace AC \rbrace,T)$} & \bfseries{$CU(\lbrace AC \rbrace,T)$} \\ \hline 
				$T_1$ & $(A:5) \, (C:10) \, (D:2)$ & 58 & 37 & 58 & 37\\
				$T_2$ & $(A:10) \, (C:6) \, (E:6) \, (G:5)$ & 97 & 62 & 97 & 62 \\ 
				$T_3$ & $(A:10) \, (B:4) \, (D:12) \, (E:6) \, (F:5)$ & 139 & 69 & 0 & 0 \\ 
				$T_4$ & $(A:5) \, (B:2) \, (C:3) \, (D:2) \, (H:2)$ & 47 & 26 & 41 & 26 \\ 
				$T_5$ & $(B:8) \, (C:13) \, (D:6) \, (E:3)$ & 99 & 58 & 0 & 0 \\ 
				$T_6$ & $(B:4) \, (C:4) \, (E:3) \, (G:2)$ & 42 & 27 & 0 & 0 \\ 
				$T_7$ & $(F:1) \, (G:2)$ & 9 & 7 & 0 & 0 \\ 
				$T_8$ & $(F:4) \, (G:3)$ & 21 & 15 & 0 & 0 \\ 
				\hline    
			\end{tabular}
			}
	\end{center}
\end{table}
\begin{figure}
	\begin{center}
	
		\includegraphics[width=6cm]{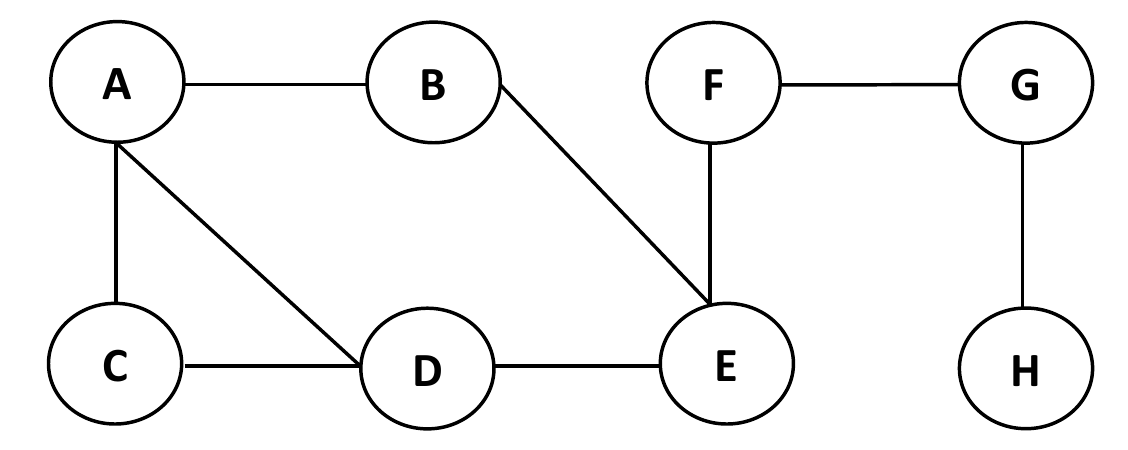}
	\end{center}
	\caption{Graph over items}
	\label{fig:graph_items}
\end{figure}
Given an undirected unweighted graph $G$ with vertices $V$, graph coverage ($Co(X)$) of a subset of vertices $X \subseteq V$ is defined as the cardinality of the set containing $X$ \footnote{$Co(X)$ returns the cardinality of a set containing $X$ to ensure that $Co(X)$ of an itemset $X$ is greater than zero.} and the immediate neighbors of the vertices in $X$. It is known that $Co(\cdot)$ is a subadditive and monotone function \cite{Krause12submodularfunction}. For example, coverage of $\lbrace A,C \rbrace$ in the graph shown in Figure~\ref{fig:Example_2} is equal to $|\{A,B,C,D\}|$ = 4.

Coming back to high-utility itemset mining, suppose in addition to a transaction database we are also given an additional graph $G$ over the items (objects) that capture their pairwise relationships (e.g., see Figure~\ref{fig:Example_2}), say follower-followee relationship. Assume that a transaction defines a set of users along with their frequency of activities on a particular day. Now we define a utility function called $ucov$ that combines the quantity information in the database and relationship among the users in terms of the coverage information in $G$. The function $ucov$ captures the notion of minimal coverage/influence of a set of active users on their immediate neighbours and mined patterns can be used as a target for marketing purposes. Functions like $ucov$ can be designed to capture sets of influential users who can be given incentives to promote a product in their immediate neighbourhood also known as Ego-network \cite{Perozzi} in the community detection literature. It can be observed that influence maximization \cite{IM} is a different problem as it aims to find the set of users that influence the maximum number of users at the end by propagating influence in multiple iterations.  

We do a comparative analysis of the patterns generated by using $ucov$ on Twitter dataset against existing baseline functions like frequency, HUIM and some of our-defined functions in Section \ref{sec:twitter}. Please note that $ucov$ is our designed example of a subadditive monotone function. The study in this paper is valid for any subadditive monotone function. We want to highlight that domain experts and pattern mining practitioners can come up with any utility function that combines information from multiple data sources to find interesting groups of items.

\begin{definition}[Coverage utility of an itemset ($ucov$)] \label{defi2}
	Let $T = \{ (x_1:q_1),$ \\ $(x_2:q_2), \ldots (x_n:q_n) \}$ be a transaction such that $\forall i \in \lbrace 1 \ldots n \rbrace$, $q_i>0$. Suppose that the items are ordered \footnote{ We order the items in a transaction T and define $ucov$ as per Definition \ref{defi2} to make the subsequent proofs for $ucov$ easier to understand.  } such that $q_i \le q_{i+1}$ for all $i$. Let $X$ be an itemset with $k$ items from T: $X = \lbrace x_1,x_2, \ldots ,x_k \rbrace$. We define the {\em coverage utility} of $X$ in $T$ in the following manner.
	$$	ucov(X,T)= q_1 \times Co(\lbrace X \rbrace) + \displaystyle\sum_{j=2}^{k} (q_j-q_{j-1}) \times Co(\lbrace x_j\cdots x_k\rbrace)$$
\end{definition}
For example, $ucov({ACD}, T_1)$ in the database in Figure~\ref{fig:Example_2} can be computed as $2 \times Co(\lbrace DAC \rbrace) + 3 \times Co(\lbrace AC \rbrace) + 5 \times Co(\lbrace C \rbrace) = 2\times 5 + 3\times4 + 5\times3 = 37$.
\begin{lemma} \label{lemma_ucov_proof}
Let $T = \{ (A_1:q_1), \ldots (A_n:q_n), (B_1:r_1), \ldots (B_m:r_m) \}$ be a transaction with positive quantity associated with every item in $T$. Suppose that the items are ordered  such that $q_1 \leq \cdots \leq q_n \leq r_1 \leq \cdots \leq r_m$. Let $X$ be an itemset with $n$ items from T: $X = \lbrace A_1,A_2, \ldots ,A_n \rbrace$. Let $Y$ be an itemset with $m$ items from T: $Y = \lbrace B_1,B_2, \ldots ,B_m \rbrace$. Then, ucov(X,T) + ucov(Y,T) $\geq$ ucov(X$\cup$Y,T).

\end{lemma}
\begin{proof}
Let us analyze the terms in ucov(X,T), ucov(Y,T), and ucov(X$\cup$Y,T). 
The first term in ucov(Y,T) is $(r_1)\times$Co($\lbrace T \rbrace$). The first term in ucov(Y,T) can be expanded such that ucov(X,T) + ucov(Y,T) $\geq$ ucov(X$\cup$Y,T) for the first n+1 terms as shown in Table \ref{Tab:proof1} as $(r_1)\times$Co($\lbrace T \rbrace$) = $(q_1)\times$Co($\lbrace T \rbrace$) + $(q_2-q_1)\times$Co($\lbrace T \rbrace$) + $\cdots$+ $(r_1-q_n)\times$Co($\lbrace T \rbrace$). The remaining terms in ucov(Y,T) and ucov(X$\cup$Y,T) are equal.  
\end{proof}

\begin{table}[!h]
	\begin{center}
		\caption{\small Comparison of ucov(X,T) + ucov(Y,T) with ucov(X$\cup$Y,T) (Lemma \ref{lemma_ucov_proof}) \label{Tab:proof1}}
		\scalebox{0.68}{
		\begin{tabular}{ p{2.7cm} p{3.2cm}  p{0.5cm} p{4cm} p{0.5cm} p{4.5cm}}
			\hline
			\bfseries{Quantity} & \bfseries{ucov(X,T)} & \bfseries{+}&\bfseries{ucov(Y,T)} & \bfseries{$\geq$}& \bfseries{ucov(X$\cup$Y,T)} \\ \hline
			$m_1=q_1$ & $m_1\times Co(\lbrace A_1\cdots A_n \rbrace)$ & + & \boldmath{$m_1\times Co(\lbrace Y \rbrace)$} & $\geq$& $m_1\times Co(\lbrace A_1\cdots A_n \cup Y  \rbrace)$\\
			
			$m_2=q_2-q_1$ & $m_2\times Co(\lbrace A_2\cdots A_n \rbrace)$ & + & \boldmath{$m_2\times Co(\lbrace Y \rbrace)$} & $\geq$ & $m_2\times Co(\lbrace A_2\cdots A_n \cup Y \rbrace)$ \\
			.& & & & & \\
			$m_{n}=q_{n}-q_{n-1}$ & $m_{n}\times Co(\lbrace A_n \rbrace)$ & + & \boldmath{$m_n\times Co(\lbrace Y \rbrace)$} & $\geq$ & $m_{n}\times Co(\lbrace A_n \cup Y \rbrace)$\\
			
			$m_{n+1}=r_1-q_{n}$ &  &  & \boldmath{$m_{n+1}\times Co(\lbrace Y \rbrace)$} & = & $m_{n+1}\times Co(\lbrace  Y \rbrace)$\\
			
			$m_{n+2}=r_2-r_1$ &  &  & $m_{n+2}\times Co(\lbrace B_2\cdots B_m \rbrace)$ & = & $m_{n+2}\times Co(\lbrace  B_2\cdots B_m \rbrace)$\\
			.& & & & & \\
			$m_{m}=r_m-r_{m-1}$ &  &  & $m_{m}\times Co(\lbrace B_m \rbrace)$ & = & $m_{m}\times Co(\lbrace  B_m \rbrace)$\\
			\hline
		\end{tabular}
		}
	\end{center}
\end{table}

\begin{table}[!h]
	\begin{center}
		\caption{\small Example (Lemma 1)
		\label{Tab:example_lemma1}}
		\scalebox{0.7}{
		\begin{tabular}{ p{2.7cm} p{3.2cm}  p{0.5cm} p{3cm} p{0.5cm} p{4.5cm}}
			\hline
			\bfseries{Quantity} & \bfseries{ucov(X,T)} & \bfseries{+}&\bfseries{ucov(Y,T)} & \bfseries{$\geq$}& \bfseries{ucov(X$\cup$Y,T)} \\ \hline
			$m_1=2$ & $2\times Co(\lbrace A_1,A_2,A_3 \rbrace)$ & + & \boldmath{$2\times Co(\lbrace B_1,B_2 \rbrace)$} & $\geq$& $2\times Co(\lbrace A_1,A_2,A_3,B_1,B_2  \rbrace)$\\
			
			$m_2=1$ & $1\times Co(\lbrace A_2,A_3 \rbrace)$ & + & \boldmath{$1\times Co(\lbrace B_1,B_2 \rbrace)$} & $\geq$ & $1\times Co(\lbrace A_2,A_3,B_1,B_2 \rbrace)$ \\
			
			$m_3=1$ & $1\times Co(\lbrace A_3 \rbrace)$ & + & \boldmath{$1\times Co(\lbrace B_1,B_2 \rbrace)$} & $\geq$ & $1\times Co(\lbrace A_3,B_1,B_2 \rbrace)$\\
			
			$m_4=1$ &  &  & \boldmath{$1\times Co(\lbrace B_1,B_2 \rbrace)$} & = & $1\times Co(\lbrace  B_1,B_2 \rbrace)$\\
			
			$m_5=3$ &  &  & $3\times Co(\lbrace B_2 \rbrace)$ & = & $3\times Co(\lbrace  B_2 \rbrace)$\\
			\hline
		\end{tabular}
		}
	\end{center}
\end{table}
For example, consider a transaction $T = \{ (A_1:2), (A_2:3), (A_3:4), (B_1:5), (B_2:8) \}$, $X = \lbrace A_1,A_2,A_3 \rbrace$, and $Y = \lbrace B_1,B_2\rbrace$ respectively. Refer Table \ref{Tab:example_lemma1} for an example for Lemma \ref{lemma_ucov_proof}. \\    

\begin{theorem}\label{thm1}
	The function ucov($\cdot$,T) is subadditive.
\end{theorem}
\begin{proof}
Let $T$ be a transaction with positive quantity associated with every item in $T$. Let the items in $T$ be sorted in ascending order of quantity. Let $X \subseteq T$ and $Y \subseteq T$. ucov(X$\cup$Y,T) will find an item with minimum quantity either from set $X$ or $Y$ in transaction $T$. ucov(X$\cup$Y,T) will find items with minimum quantity from one set X( or Y) followed by an item from the other set Y(or X) respectively. Using Lemma 1, the terms in ucov(X,T) and ucov(Y,T) can be expanded such that ucov(X,T) + ucov(Y,T) $\geq$ ucov(X$\cup$Y,T). Hence, ucov($\cdot$,T) is a subadditive function.
\end{proof}

For example, consider a transaction $T = \{ (A_1:2), (A_2:3), (B_1:4), (A_3:7), (B_2:8), (A_4:9), (B_3:10)  \}$, $X = \lbrace A_1,A_2,A_3, A_4 \rbrace$, and $Y = \lbrace B_1,B_2,B_3\rbrace$ respectively. Refer Table \ref{Tab:example_theorem1} for an example for Theorem \ref{thm1}.   
\begin{table}
	\begin{center}
		\caption{\small Example (Theorem \ref{thm1})
		\label{Tab:example_theorem1}}
		\scalebox{0.65}{
		\begin{tabular}{ p{1.5cm} p{4cm}  p{0.5cm} p{3.5cm} p{0.5cm} p{6cm}}
			\hline
			\bfseries{Quantity} & \bfseries{ucov(X,T)} & \bfseries{+}&\bfseries{ucov(Y,T)} & \bfseries{$\geq$}& \bfseries{ucov(X$\cup$Y,T)} \\ \hline
			$m_1=2$ & $2\times Co(\lbrace A_1,A_2,A_3,A_4 \rbrace)$ & + & \boldmath{$2\times Co(\lbrace B_1,B_2,B_3 \rbrace)$} & $\geq$& $2\times Co(\lbrace A_1,A_2,A_3,A_4,B_1,B_2,B_3  \rbrace)$\\
			
			$m_2=1$ & $1\times Co(\lbrace A_2,A_3,A_4 \rbrace)$ & + & \boldmath{$1\times Co(\lbrace B_1,B_2,B_3 \rbrace)$} & $\geq$ & $1\times Co(\lbrace A_2,A_3,A_4,B_1,B_2,B_3 \rbrace)$ \\
			
			$m_3=1$ & \boldmath{$1\times Co(\lbrace A_3,A_4 \rbrace)$} & + & \boldmath{$1\times Co(\lbrace B_1,B_2,B_3 \rbrace)$} & $\geq$ & $1\times Co(\lbrace A_3,A_4,B_1,B_2,B_3 \rbrace)$\\
			
			$m_4=3$ & \boldmath{$3\times Co(\lbrace A_3,A_4 \rbrace)$} & + & \boldmath{$3\times Co(\lbrace B_2,B_3 \rbrace)$} & $\geq$ & $3\times Co(\lbrace A_3,A_4,B_2,B_3 \rbrace)$\\
			
			$m_5=1$ & \boldmath{$1\times Co(\lbrace A_4 \rbrace)$} & + & \boldmath{$1\times Co(\lbrace B_2,B_3 \rbrace)$} & $\geq$ & $1\times Co(\lbrace A_4,B_2,B_3 \rbrace)$\\
			
			$m_6=1$ & \boldmath{$1\times Co(\lbrace A_4 \rbrace)$} & + & \boldmath{$1\times Co(\lbrace B_3 \rbrace)$} & $\geq$ & $1\times Co(\lbrace A_4,B_3 \rbrace)$\\
			
			$m_7=1$ &  &  & $1\times Co(\lbrace B_3 \rbrace)$ & = & $1\times Co(\lbrace  B_3 \rbrace)$\\
			\hline
		\end{tabular}
		}
	\end{center}
\end{table}

\begin{theorem}\label{thm2}
	ucov($\cdot$,T) is a monotone function.
\end{theorem}
\begin{proof}
Let $T = \{ (A_1:q_1), \ldots (A_n:q_n), (B_1:r_1) \}$ be a transaction with positive quantity associated with every item in $T$. Suppose that the items are ordered  such that $q_1 \leq \cdots \leq q_n$ and $q_{i-1} \leq r_1 \leq q_i$. Let $X$ be an itemset with $n$ items from T: $X = \lbrace A_1,A_2, \ldots ,A_n \rbrace$. Let $Y$ be an itemset from T: $Y = \lbrace B_1 \rbrace$. Let us analyze the terms ucov(X,T) and ucov(X$\cup$Y,T) as shown in Table \ref{Tab:proof2}.  It can be quickly verified that the terms from 1st till $i-1^{th}$ rows in Table \ref{Tab:proof2} from ucov(X,T) is less than or equal to ucov(X$\cup$Y,T) as $Co(\cdot)$ is monotone. Let us consider the $i^{th}$ term in ucov(X$\cup$Y,T). It can be observed that the next term from ucov(X,T) can be expanded such that ucov(X,T) $\leq$ ucov(X$\cup$Y,T) as $(q_i-q_{i-1})\times$Co($\lbrace A_i\cdots A_n \rbrace$) = $(q_i-r_1)\times$Co($\lbrace A_i\cdots A_n \rbrace$) +$(r_1-q_{i-1})\times$Co($\lbrace A_i\cdots A_n \rbrace$). Hence, ucov($\cdot$,T) is a monotone function.

\end{proof}
\begin{table}[!h]
	\begin{center}
		\caption{\small Comparison of ucov(X,T) with ucov(X$\cup$Y,T) (Theorem \ref{thm2})\label{Tab:proof2}}
		\scalebox{0.7}{
		\begin{tabular}{ p{3.7cm} p{4.5cm}  p{0.5cm} p{4.5cm}}
			\hline
			\bfseries{Quantity} & \bfseries{ucov(X,T)} & \bfseries{$\leq$}& \bfseries{ucov(X$\cup$Y,T)} \\ \hline
			$m_1=q_1$ & $m_1\times Co(\lbrace A_1\cdots A_n \rbrace)$ & $\leq$& $m_1\times Co(\lbrace A_1\cdots A_n \cup Y  \rbrace)$\\
			
			$m_2=q_2-q_1$ & $m_2\times Co(\lbrace A_2\cdots A_n \rbrace)$ & $\leq$ & $m_2\times Co(\lbrace A_2\cdots A_n \cup Y \rbrace)$ \\
			.& & & \\
			$m_{i-1}=q_{i-1}-q_{i-2}$ & $m_{i-1}\times Co(\lbrace A_{i-1}\cdots A_n \rbrace)$ & $\leq$ & $m_{i-1}\times Co(\lbrace A_{i-1}\cdots A_n \cup Y \rbrace)$\\
			
			$m_i=r_1-q_{i-1}$ & \boldmath{$m_i\times Co(\lbrace A_i\cdots A_n  \rbrace)$} & $\leq$ & $m_i\times Co(\lbrace A_i\cdots A_n \cup Y \rbrace)$\\
			
			$m_{i+1}=q_i-r_1$ & \boldmath{$m_{i+1}\times Co(\lbrace A_i\cdots A_n \rbrace)$} & = & $m_{i+1}\times Co(\lbrace A_i\cdots A_n \rbrace)$\\
			
			$m_{i+2}=q_{i+1}-q_i$ & $m_{i+2}\times Co(\lbrace A_{i+1}\cdots A_n \rbrace)$ & = & $m_{i+2}\times Co(\lbrace A_{i+1}\cdots A_n\rbrace)$\\ 
			.& & & \\
			$m_{n+1}=q_n-q_{n-1}$ & $m_{n+1}\times Co(\lbrace A_n \rbrace)$ & = & $m_{n+1}\times Co(\lbrace A_n\rbrace)$\\
			\hline
		\end{tabular}
		}
	\end{center}
\end{table}

\section{Algorithms for HUIM-SM}
\label{sec:hguimalgo}
A superset of a low-utility itemset can have high-utility as utility functions do not necessarily follow the anti-monotonicity property. Hence to avoid searching all itemsets, HUIM algorithms define bounds such as transaction utility(TU), transaction weighted utility(TWU) \cite{twdc}, exact-utility and remaining-utility \cite{hui_miner}. We now discuss these bounds in the context of a general subadditive and monotone utility function $u(X,T)$.

\subsection{TU and TWU bounds}\label{subsec:tu}
The $TU$ of a transaction is the {\em sum} of the utilities of its items. $TWU$ of an itemset $X$ is the sum of $TU$ for transactions that contains $X$. The {\em transaction-weighted downward closure} property of $TWU$ ensures that if $TWU(X)$ is less than the threshold, $X$ and its super-sets cannot be high-utility itemsets \cite{twdc}. We now define a related concept of transaction subadditive-monotone utility.

\begin{definition}[TSMU and TSMWU]
	The transaction subadditive monotone utility of a transaction $T$ is defined as $TSMU(T)=u(T,T)$. The transaction subadditive monotone weighted utility (TSMWU) of an itemset $X$ is the sum of $TSMU(T)$ for all the transactions containing $X$. 
\end{definition}
For example, consider transaction $T_1$ from Table \ref{fig:Example_2}. TU for $T_1$ for $ucov$ is equal to $ucov(\lbrace A \rbrace,T_1)+ucov(\lbrace C \rbrace,T_1)+ucov(\lbrace D \rbrace,T_1)$ which evaluates to 20 + 30 + 8 = 58. TSMU for $T_1$ for $ucov$ is equal to $ucov(\lbrace ACD \rbrace, T_1)$ which evaluates to 37.

Now we show that $TSMWU$ satisfies the downward-closure property making it useful in mining algorithms and a tighter bound compared to TWU.
\begin{lemma} \label{tsmwu_prune_lemma}
If TSMWU(X) is less than the threshold, X and its super-sets cannot be high-utility itemsets.
\end{lemma}
\begin{proof}
For any transaction $T$ and any $X \subseteq X' \subseteq T$, $u(X',T) \leq TSMU(X)$ as the set of transactions containing $X'$ is always a subset of the transactions containing $X$ and $u(\cdot, T)$ is a monotone function. Therefore, $u(X') \leq TSMWU(X)$ for all $X \subseteq X'$. 
\end{proof}

\begin{lemma} \label{tsmwu_tight_lemma}
$TSMWU(X)$ is a tighter upper-bound compared to $TWU(X)$.
\end{lemma}
\begin{proof}
This can be easily proved since $TU(X) = \sum_{T \in D} \sum_{x \in T} u(x,T)$ \\ $\ge \sum_{T \in D} u(T,T)$ = $TSMU(X)$ due to subadditivity. Therefore, $TSMWU(X) \le TWU(X)$.
\end{proof}

\subsection{Exact-utility (EU) and Remaining-utility (RU) bounds}\label{subsec:euru}
List-based \cite{hui_miner,fhm} and projection-based \cite{efim,d2hup} HUIM algorithms use the Exact-utility (EU) and Remaining-utility (RU) bounds that are tighter compared to the TU bound explained earlier. These algorithms process items according to some fixed order and items in each transaction are sorted accordingly; let $X$ be some itemset that is being processed, $T$ be some transaction containing $X$ and $T/X$ be the {\em items appearing after $X$} in $T$. One should note that $X \cup (T/X)$ is not $T$, in particular, items appearing in $T$ {\em before} $X$ in the processing order are not in $T/X$. The exact utility $EU(X,T)$ is the {\em sum} of the utilities of the items in $X$ in $T$ and the remaining utility $RU(X,T)$ is defined as the {\em sum} of the utilities of the items in $T/X$. It can be shown that $X$ and its extensions (according to the processing order) cannot be high utility if $EU_-RU(X) = \sum_{T \supseteq X} EU(X,T) + RU(X,T)$ is less than the threshold \cite{hui_miner}; this fact is used in HUIM algorithms to decide whether to examine extensions of the currently explored itemset $X$. Now, we define combined utility in an analogous manner for subadditive monotone utility functions and analyze its applicability for itemset mining.
\begin{definition}[Combined utility ($CU$)]
    Given an itemset $X$ and a transaction $T$ with $X \subseteq T$, the combined utility of $X$ is defined as $CU(X,T) = u(X \cup T/X,T)$ and $CU(X) = \sum_{T \supseteq X} CU(X,T)$.
\end{definition}
For example, consider transaction $T_1$ from Table \ref{fig:Example_2}. Let us compute $EU(\lbrace AC\rbrace,T_1)$ \\ + $RU(\lbrace AC\rbrace,T_1)$, and $CU(\lbrace AC \rbrace,T_1)$. $EU(\lbrace AC\rbrace,T_1)+RU(\lbrace AC\rbrace,T_1)$ evaluates to 58, and $CU(\lbrace AC \rbrace,T_1)$ evaluates to 37.  

\begin{lemma} \label{cu_prune_lemma}
If $CU(X)$ is less than a threshold, any extension $X'$ of $X$ in the processing order of items is not a high-utility itemset.
\end{lemma}
\begin{proof}
Let $X'$ be some extension of $X$ in the processing order of items; since transactions are also (implicitly) stored in the same order, $X' \subseteq X \cup T/X$. The proof of the lemma follows easily since the set of transactions containing $X'$ is always a subset of the set of transactions containing $X$ and $u(\cdot, T)$ is a monotone function which implies that $u(X',T) \le u(X \cup T/X,T) = CU(X,T)$.
\end{proof}

\begin{lemma} \label{cu_tight_lemma}
$CU(X)$ is a tighter upper-bound compared to $EU_-RU(X)$.
\end{lemma}
\begin{proof}
First observe that subadditivity implies that $EU(X,T) = \sum_{x \in X} u(x,T)$ \\ $\ge u(X,T)$. Similarly it follows that $RU(X,T) \ge u(T/X,T)$. Therefore, we immediately get $EU(X,T) + RU(X,T) \ge u(X,T) + u(T/X,T) \ge u(X \cup T/X,T) = CU(X,T)$ (second inequality is again due to subadditivity). Therefore it follows that $EU_-RU(X) \ge CU(X)$.
\end{proof}
It should be noted that $CU(X)=EU_-RU(X)$ when $Sum$ (defined in Section~\ref{sec:our}) is used as the utility function. The results of Subsections~\ref{subsec:tu} and \ref{subsec:euru} show that we can use $TSMWU$ in place of $TWU$ and $CU$ in place of $EU_-RU$ in HUIM algorithms. We illustrate these bounds in our example transaction database presented in Table~\ref{fig:Example_2} for $ucov$. In the following subsection, we design a new data structure called SMI-list and an algorithm SM-Miner. We also show how to adapt existing tree-based and projection-based HUIM algorithms to our general setting and highlight important changes that must be incorporated in them.

\subsection{List-based algorithm}
List-based algorithms \cite{hui_miner,fhm} for HUIM maintains a utility-list consisting of tuples (Transaction id (Tid), Exact-utility (EU(X,T)), Remaining-utility (RU(X,T))) with each itemset X. List-based algorithms scan the transaction database to construct the utility-list for promising items i.e. items with TWU no less than the minimum utility threshold. The utility-list for a $\lbrace k \rbrace$-itemset (an itemset consisting of k items) is constructed intersecting lists of two $\lbrace k-1 \rbrace$ itemsets and the prefix itemset. Two variables sumEU, and sumRU store the sum of EU(X,T), and RU(X,T) for all transactions which contain itemset X. 

We propose a data structure called SMI-list to store tuples of the form (Transaction Id (Tid), Current Weighted Itemset (CWI), Remaining Weighted Itemset (RWI)). The Current Weighted Itemset of an SMI-list for an itemset X stores the item-quantity information for all items in X present in a transaction. The Remaining Weighted Itemset stores the items with their quantities which appear after X in a transaction. The variables SumEU, and SumCU accumulate the EU and CU value during the construction of the inverted-list.  A variable Combined utility(CU) stores the Combined Utility (CU(X)) bound with the utility-list for every itemset X. The construction process is shown in Algo \ref{algo_ds}. We observe that there is no need to scan the SMI-list of the prefix while constructing the list for a $k$-itemset, unlike required by some algorithms like HUI-Miner, FHM \cite{hui_miner,fhm} for HUIM.   

\begin{algorithm}
	\caption{Construct-SMI-List ($Ix$,$Iy$,$f(\cdot)$)}
	\textbf{Input:} $Ix$: SMI-List of $Ix$, $Iy$: SMI-List of $Iy$, $f(\cdot)$: subadditive monotone utility function.\\
	\textbf{Output:} SMI-List of $Ixy$.
	\label{algo_ds}
	\begin{algorithmic}[1]
	    \State $Ixy=\emptyset$
		\For{each element $Ex$ in $Ix$}
			\If {$\exists Ey \in Iy$ and $Ex.Tid==Ey.Tid$} 
				\State $Exy=\langle Ex.Tid, Ex.CWI \cup Ey.CWI, Ey.RWI \rangle$ 
				\State Append $Exy$ to $Ixy$
				\State $Ixy.sumEU+=f(Ex.CWI \cup Ey.CWI)$
				\State $Ixy.sumCU+=f(Ex.CWI \cup Ey.CWI \cup Ey.RWI)$
			\EndIf
		\EndFor
		\State Return $Ixy$
	\end{algorithmic}
\end{algorithm}

\begin{table}
	
	\begin{center}
		\caption{SMI-List of $\lbrace A \rbrace$ \label{Tab:A}}
		\begin{tabular}{ l l l}
			\hline
			\bfseries{TID} & \bfseries{CWI} & \bfseries{RWI}   \\ \hline
			1 & $\{ (A:5)\}$ & $\{ (C:10),(D:2)\}$\\
			2 & $\{ (A:10)\}$ & $\{ (C:6),(E:6),(G:5)\}$\\
			3 & $\{ (A:10)\}$ & $\{ (B:4),(D:12),(E:6),(F:5)\}$\\
			4 & $\{ (A:5)\}$ & $\{ (B:2),(C:3),(D:2),(H:2)\}$\\ \hline
				
		\end{tabular}
	\end{center}
\end{table}

\begin{table}
	
	\begin{center}
		\caption{SMI-List of $\lbrace B \rbrace$ \label{Tab:B}}
		\begin{tabular}{ l l l}
			\hline
			\bfseries{TID} & \bfseries{CWI} & \bfseries{RWI}   \\ \hline
			3 & $\{ (B:4)\}$ & $\{ (D:12),(E:6),(F:5)\}$\\
			4 & $\{ (B:2)\}$ & $\{ (C:3),(D:2),(H:2)\}$\\
			5 & $\{ (B:8)\}$ & $\{ (C:13),(D:6),(E:3)\}$\\
			6 & $\{ (B:4)\}$ & $\{ (C:4),(E:3),(G:2)\}$\\
				\hline	
		\end{tabular}
	\end{center}
\end{table}

\begin{table}
	
	\begin{center}
		\caption{SMI-List of $\lbrace AB \rbrace$ \label{Tab:AB}}
		\begin{tabular}{ l l l}
			\hline
			\bfseries{TID} & \bfseries{CWI} & \bfseries{RWI}   \\ \hline
			3 & $\{ (A:10),(B:4)\}$ & $\{ (D:12),(E:6),(F:5)\}$\\
			4 & $\{ (A:5),(B:2)\}$ & $\{ (C:3),(D:2),(H:2)\}$\\ \hline
				
		\end{tabular}
	\end{center}
\end{table}

Our proposed algorithm called SM-Miner ( Algorithm \ref{algo_vert}) takes as input a SMI-List  of a prefix (say $\alpha$), SMI-List of $\alpha$ 1-extensions, minimum utility threshold ($\theta$), and a subadditive monotone utility function ($f(\cdot)$). Initially, the database is scanned to compute the TSMWU of the items present in the transaction database as TSMWU is a tighter bound then TWU as per Lemma \ref{tsmwu_tight_lemma}. Another database scan is performed to remove items with TSMWU less than $\theta$ (Lemma \ref{tsmwu_prune_lemma}) and construct the SMI-List of the remaining items present in the transaction database. For example, the SMI-List of $\lbrace A \rbrace$, and $\lbrace B \rbrace$ is shown in Table \ref{Tab:A}, and Table \ref{Tab:B} respectively. SM-Miner explores the search space in a depth-first search manner and returns the complete set of high-utility itemsets. SM-Miner constructs the SMI-List of a $\lbrace k \rbrace$-itemsets from SMI-List of two $\lbrace k-1 \rbrace$-itemsets. As an example, the SMI-List of $\lbrace AB \rbrace$ is shown in Table \ref{Tab:AB}. Two variables SumEU, and SumCU accumulate the EU and CU values during the construction of the SMI-List of an itemset as CU is a tighter bound (Lemma \ref{cu_tight_lemma}). If SumCU for an itemset X is less than $\theta$, X and its supersets can't be high-utility as per Lemma \ref{cu_prune_lemma}.   
\begin{algorithm}
	\caption{SM-Miner (SMI-List of $\alpha$ ,$Ext\_I$,$\theta$,$f(\cdot)$)}
	\textbf{Input:} SMI-List of prefix $\alpha$ (initially empty) ,$Ext\_I$: SMI-List of $\alpha$ 1-extensions, $\theta$: a user-specified threshold, $f(\cdot)$: subadditive monotone utility function.\\
	\textbf{Output:} the set of high-utility itemsets with $\alpha$ as prefix.
	\label{algo_vert}
	\begin{algorithmic}[1]
		\For{each SMI-List $Ix$ in $Ext\_I$}
			\If {$Ix.sumEU \geq \theta$} 
				\State Ix is a high-utility itemset
			\EndIf
			\If {$\, Ix.sumCU \geq \theta$} \Comment{Estimate the utility of supersets of Ix }
				\State $Ext\_Ix$=$\emptyset$
				\For{each SMI-List $Iy$ after $Ix$ in $Ext\_I$}
					
					\State $Ixy$=Construct-SMI-List($Ix$,$Iy$,$f(\cdot)$)
					\State $Ext\_Ix$=$Ext\_Ix \cup Ixy$
				\EndFor
				\State SM-Miner ($Ix$,$Ext\_Ix$,$\theta$,$f(\cdot)$)
			\EndIf
		\EndFor
		
	\end{algorithmic}
\end{algorithm}

\subsection{Tree-based algorithm}
A tree-based algorithm like UP-Growth and UPGrowth+ \cite{up_tree} creates a tree data structure from transaction database which is then used to get potential high utility candidates. Every node of tree stores the item name, support, upper-bound utility estimate like TWU, pointer to the parent node, and a set of child nodes. Tree-based algorithms generate potential candidates during the first phase. In the next phase, a database scan is performed to find the exact high-utility itemsets by computing the exact utility value for the generated candidates. Tree-based algorithms create a global tree structure before starting the mining process. Local trees are recursively created from the global tree only to generate candidates. An itemset $X$ is unpromising in a transaction database if $TWU(X)$ is less than the minimum utility threshold $\theta$. Tree-based algorithms remove unpromising items during global and local tree creation to compute better utility estimates. We observe that unpromising items can be removed during global tree creation as it's possible to recompute the utility of a transaction after removing the unpromising items for any arbitrary subadditive function. However, we observe that removing unpromising items during local tree creation may not give correct upper-bound estimates for any arbitrary subadditive function. 

Let $T = \{ (A_1:q_1), \ldots (A_n:q_n) \}$ be a transaction with positive quantity associated with every item in $T$. Let $X$ be an itemset from T: $X = \lbrace (A_1) \rbrace$. Let $Y$ be an itemset with ($n-1$) items from T: $Y = \lbrace (A_2, \ldots ,A_n) \rbrace$. Let us assume that the item $A_1$ is unpromising i.e. $TSMWU(A_1)$ is less than the minimum utility threshold $\theta$. Item $A_1$ can't be a part of any high-utility itemset. It can be easily verified that $u(X,T) + u(Y,T) = u(X\cup Y,T)$ when u($\cdot$) is the $Sum(\cdot)$ defined for HUIM. Therefore, $u(Y,T) = u(X \cup Y,T) - u(X,T)$ gives a correct bound for HUIM. However, computing tighter utility estimates by removing unpromising items can result in incorrect upper bound estimate i.e. resulting in false negatives for those functions where $f(X,T) + f(Y,T) > f(X\cup Y,T)$. 

Tree-based algorithms for HUIM can be adapted with minimal change for an arbitrary subadditive monotone function. There is no change in the tree data structure, and removing unpromising items during local tree creation must be disabled during the candidate generation phase. The tree construction and verification phase remain the same. This leads to direct use of utility values associated with the tree nodes in the mining process. However, the absence of removing unpromising items during local tree creation can result in the generation of a large set of potential candidates. 

\subsection{Projection-based algorithm}
 Projection-based algorithms like EFIM \cite{efim} merge two identical transactions to reduce the cost for database scan during the construction of the projected database for an itemset.
 
 Let $T = \{ (A_1:q_1), \ldots (A_n:q_n) \}$ be a transaction with positive quantity associated with every item in $T$. Let $S = \{ (A_1:r_1), \ldots (A_n:r_n) \}$ be another transaction with positive quantity associated with every item in $S$. Let $M = \{ (A_1:q_1+r_1), \ldots (A_n:q_n+r_n) \}$ be another transaction with positive quantity associated with every item in $M$. Merging transactions does not change the utility of the itemset $\langle A_1,A_2,....A_n \rangle$ in the database for the $Sum$ function defined by HUIM. However, utility of an itemset in the merged transaction can be less than the sum of its utility in individual transactions for an arbitrary subadditive function. Consider the transactions $T_7$ and $T_8$ in Figure \ref{fig:Example_2}. The utility ($ucov$) of the itemset $\lbrace FG \rbrace$ in $T_7$ and $T_8$ is 7 and 15 respectively. Consider merging transactions $T_7$ and $T_8$ to a single transaction $M = \{ (F:5), \ldots (G:5) \}$. The utility of itemset $\lbrace FG \rbrace$ in this merged transaction i.e. u($\lbrace FG \rbrace$,M) is 20. Therefore, transaction merging can change the utility of an itemset in the database. We disable transaction merging when adapting projection-based algorithms like EFIM for an arbitrary subadditive monotone function. 
 
 \section{Case study of HUIM-SM with coverage functions} \label{sec:twitter}
 We conduct an experimental study on a publicly available Twitter dataset to emphasize that utility functions like $ucov$ can be designed that can combine information from a transaction database and domain knowledge from an external data source in the form of a graph. We extract the top-100 patterns extracted by $ucov$ and other mining methodologies like frequent itemset mining and high-utility itemset mining. The Twitter-Dynamic-Net dataset\footnote{The dataset called Twitter Dynamic Net is publicly available as a part of the AMiner repository.( \url{https://aminer.org/data-sna})} was constructed by selecting "Lady Gaga" on Twitter, and randomly selecting 10,000 of her followers. The selected followers were taken as seed users, and all their followers were collected by a crawler. Every tweet has information like user name, tweet id, timestamp, and retweet by user name, etc. fields associated with it. We construct a transaction database and a directed follower-followee graph from the dataset. Every transaction captures the users who were active during the period associated with the transaction. The quantity associated with a user is the number of tweets posted by that user in the time interval associated with the transaction. The follower-followee graph has users as vertices, and a directed edge from a user 'A' to user 'B' if 'B' has retweeted at least one tweet posted by user 'A' in the complete dataset. We fix the time duration to three hours and construct a transaction database whose statistics are given in Table \ref{Tab:twitter}. The followee-follower graph constructed from the dataset contains 14,766 followees and 20,423 followers with 46,164 edges between followees and followers. The average degree of a followee is 3, and the maximum degree is 48. Every user in the constructed graph has at least one follower. 
\begin{table}
	
	\begin{center}
		\caption{Characteristics  of  Twitter  transaction  dataset \label{Tab:twitter}}
		\begin{tabular}{ l l  l l l}
			\hline
			\bfseries{Dataset} & \bfseries{\#Tx} & \bfseries{Avg. length}& \bfseries{\#Items} & \bfseries{Max. length }   \\ \hline
			Twitter & 2,631 & 87 & 14,766 & 381 \\ \hline
		\end{tabular}
	\end{center}
\end{table}
 
We compare the top-100 patterns extracted by $ucov$ with the existing baselines: frequent itemset mining ($fim$), and HUIM ($sum$). We also define two new baseline functions by integrating domain knowledge captured in the form a graph with frequent itemset mining, and HUIM. We call our-defined baseline as frequent itemset mining with coverage ($fcov$), and HUIM with coverage ($sumcov$) respectively. 

Frequent itemset mining ($fim$) is applied to the transaction dataset only to extract top-k frequent itemsets. $fcov$ for an itemset returns the frequency of the itemset in the transaction database multiplied by its graph coverage $Co(\cdot)$. $sum$ function captures the high-utility itemsets from the transaction database only. To integrate the coverage graph in the existing framework of high-utility itemset mining, we define the utility of an item in a transaction as the product of its quantity and coverage from the directed follower-followee graph. The utility of an itemset in a transaction is simply the sum of the utility of its items. We call this function as $sumcov$. It can be observed that $fim(\cdot,T)$, $fcov(\cdot,T)$, $sum(\cdot,T)$, and $sumcov(\cdot,T)$ for a transaction T are a subadditive monotone functions.  

We compare the overlap among top-100 patterns generated by $ucov$ with $fim$, $fcov$, $sum$ and $sumcov$ in Table \ref{Tab:ucov}. It can be observed that $ucov$ extracts different patterns compared to $fim$, and $sum$ functions. There is a significant overlap in the top-100 patterns generated by $sumcov$, and $ucov$. We claim that the patterns generated by $ucov$ will always be a subset of the patterns generated by $sumcov$ for the same minimum utility threshold. We present a formal proof for our claim in Theorem \ref{thm_sumcov}. We observed empirically that $sumcov$ generates a large number of patterns compared by $ucov$ for a fixed utility threshold. For example, $sumcov$ generated 10,29,921 patterns while $ucov$ generated only 681 patterns when the minimum utility threshold was set to 10,000 for the Twitter dataset used in our experimental study.

We also analyze the pattern length and $Co(\cdot)$ for the top-100 patterns generated by the different utility functions. Table \ref{Tab:length_stats} shows the distribution of pattern length i.e. number of items contained in a pattern by different functions for the top-100 patterns.  It can be observed that $fim$ generates patterns containing one item only. 

\begin{table}
	
	\begin{center}
		\caption{\small Overlap between the top 100 patterns  generated by $fim$, $fcov$, $sum$, and $sumcov$ with $ucov$.  \label{Tab:ucov}}
		\begin{tabular}{ l l l}
			\hline
			\bfseries{Function} & \bfseries{Total number of patterns} & \bfseries{Patterns common with $ucov$}\\ \hline
			$fim$ & 100 & 21\\
			$fcov$ & 100 & 37 \\
			$sum$ & 100 & 16\\
			$sumcov$ & 100 & 71\\
			 \hline
		\end{tabular}
	\end{center}
\end{table}

\begin{table}
	
	\begin{center}
		\caption{\small Statistics showing  distribution of pattern length for top 100 patterns generated by $fim$, $fcov$, $sum$, $sumcov$, and $ucov$.    \label{Tab:length_stats}}
		\begin{tabular}{ l l  l l  l l }
			\hline
			\bfseries{Statistics} & \bfseries{$fim$} & \bfseries{$fcov$}& \bfseries{$sum$} &  \bfseries{$sumcov$} &  \bfseries{$ucov$}\\ \hline
			Minimum & 1 & 1 & 1 & 1 & 1   \\
			Maximum & 1 & 2 & 316 & 7 & 5  \\
			Mean & 1 & 1.4 & 249 & 2.5 & 2.06   \\
			Median & 1 & 1 & 315 & 2 & 2  \\
			
			 \hline
		\end{tabular}
	\end{center}
\end{table}

\begin{table}[!h]
	
	\begin{center}
		\caption{\small Statistics showing  distribution of $Co(\cdot)$ for top 100 patterns generated by $fim$, $fcov$, $sum$, $sumcov$, and $ucov$.  \label{Tab:co_stats}}
		\begin{tabular}{ l l  l l  l l }
			\hline
			\bfseries{Statistics} & \bfseries{$fim$} & \bfseries{$fcov$}& \bfseries{$sum$} &  \bfseries{$sumcov$} &  \bfseries{$ucov$} \\ \hline
			Minimum & 1 & 12 & 7 & 16 &  16  \\
			Maximum & 40 & 64 & 1798 & 105 & 79  \\
			Mean & 14.97 & 30.62 & 1421 & 42.7 & 38.9  \\
			Median & 15.5 & 29 & 1793 & 38 & 38 \\
			
			 \hline
		\end{tabular}
	\end{center}
\end{table}

\begin{table}[!h]
	
	\begin{center}
		\caption{\small Statistics showing  distribution of $f(\cdot)$ for top 100 patterns generated by $fim$, $fcov$, $sum$, $sumcov$, and $ucov$.  \label{Tab:f_stats}}
		\begin{tabular}{ l l  l l  l l }
			\hline
			\bfseries{Statistics} & \bfseries{$fim$} & \bfseries{$fcov$}& \bfseries{$sum$} &  \bfseries{$sumcov$} &  \bfseries{$ucov$} \\ \hline
			Minimum & 179 & 68 & 1 & 8 & 15  \\
			Maximum & 482 & 462 & 462 & 462 & 462  \\
			Mean & 231 & 181 & 40.9 & 99.6 & 115  \\
			Median & 213 & 170 & 1 & 67.5 & 73 \\
			
			 \hline
		\end{tabular}
	\end{center}
\end{table}

The $sum$ function is applied on the transaction database only and it generates very long patterns with a length greater than 300. However, our preference is to find patterns with a reasonable length as it's infeasible to give a discount to many active users so that they influence their set of followers for applications like viral marketing. Table \ref{Tab:co_stats} shows the distribution of graph coverage $Co(\cdot)$ for different functions. Frequent itemset mining and high-utility itemset mining ignores the domain knowledge captured in the form of a graph completely and only considers the transaction database to generate patterns. It can observed that the patterns generated by $ucov$ have better $Co(\cdot)$ compared to the patterns generated by $fim$, and $sum$. Table \ref{Tab:f_stats} shows the distribution of the frequency i.e. the number of transactions for the top-100 patterns generated by different functions. The $sum$ function generates patterns with very low frequency as longer patterns are usually present in few transactions in a transaction database. The frequent itemset mining functions generate patterns that have frequency higher than $sum$, $sumcov$ and $ucov$.

\begin{theorem}\label{thm_sumcov}
	For a fixed minimum user-defined threshold $\theta$, the set of high-utility patterns generated by $ucov(\cdot)$ will always be a subset of the patterns generated by $sumcov(\cdot)$.
\end{theorem}
\begin{proof}
Let $T = \{ (A_1:q_1), \ldots (A_n:q_n) \}$ be a transaction with positive quantity associated with every item in $T$. Suppose that the items are ordered  such that $q_1 \leq \cdots \leq q_n$. Let $X$ be an itemset with $n$ items from T: $X = \lbrace A_1,A_2, \ldots ,A_n \rbrace$. Let's analyze the functions $ucov(X,T)$ and $sumcov(X,T)$.  $ucov(X,T)$ = $(q_1)\times$Co($X$) + $(q_2-q_1)\times$Co($\lbrace A_2\cdots A_n \rbrace$)+ $\ldots$ +$(q_n)\times$Co($\lbrace A_n \rbrace$). $sumcov(X,T)$ = $(q_1)\times$Co($\lbrace A_1 \rbrace$) + $(q_2)\times$Co($\lbrace A_2 \rbrace$) + $(q_3)\times$Co($\lbrace A_3 \rbrace$) + \\ $\ldots$ +$(q_n)\times$Co($\lbrace A_n \rbrace$). It can be quickly verified that the first term in $ucov(X,T)$ i.e. $(q_1)\times$Co($\lbrace A_1\cdots A_n \rbrace$) $\leq$ $(q_1)\times$ (Co($\lbrace A_1 \rbrace$) + Co($\lbrace A_2 \rbrace$)+...+Co($\lbrace A_n \rbrace$)). Similarly, $(q_2-q_1)\times$Co($\lbrace A_2\cdots A_n \rbrace$) $\leq$ $(q_2-q_1)\times$ (Co($\lbrace A_2 \rbrace$) + Co($\lbrace A_3 \rbrace$)+ \\ $\ldots$+Co($\lbrace A_n \rbrace$)). The terms in $sumcov(X,T)$ can be expanded such that $ucov(X,T)$ is less than or equal to $sumcov(X,T)$ as shown in Table \ref{Tab:proof3}. It proves that $sumcov(\cdot)$ can assign an equal or higher score to the patterns compared to $ucov(\cdot)$. All patterns with $ucov(\cdot)$ no less than $\theta$ will have $sumcov(\cdot)$ no less than $\theta$ too. There can be extra patterns generated in the result set by $sumcov(\cdot)$ for a fixed $\theta$ compared to the patterns generated by $ucov(\cdot)$ as $sumcov(\cdot)$ can overestimate the coverage of patterns if the set of vertices covered by individual items present in the pattern overlap.
\end{proof}

\begin{table}
	\begin{center}
		\caption{\small Comparison of ucov(X,T) with sumcov(X,T) (Theorem \ref{thm_sumcov})\label{Tab:proof3}}
		\scalebox{0.68}{
		\begin{tabular}{ p{3.5cm} p{4.5cm}  p{0.5cm} p{8cm}}
			\hline
			\bfseries{Quantity} & \bfseries{ucov(X,T)} & \bfseries{$\leq$}& \bfseries{sumcov(X,T)} \\ \hline
			$m_1=q_1$ & $m_1\times$Co($\lbrace A_1\cdots A_n \rbrace$) & $\leq$& $m_1\times$ ( Co($\lbrace A_1 \rbrace$) + Co($\lbrace A_2 \rbrace$) +$\ldots$+Co($\lbrace A_n \rbrace$)  )\\
			
			$m_2=q_2-q_1$ & $m_2\times$Co($\lbrace A_2\cdots A_n \rbrace$) & $\leq$ & $m_2\times$ ( Co($\lbrace A_2 \rbrace$) + Co($\lbrace A_3 \rbrace$) +$\ldots$+Co($\lbrace A_n \rbrace$)  ) \\
			.& & & \\
			.& & & \\
			.& & & \\
			$m_{n}=q_{n}-q_{n-1}$ & $m_{n}\times$Co($\lbrace A_n \rbrace$) & = & $m_{n}\times$Co($\lbrace A_n \rbrace$)\\
			
			\hline
		\end{tabular}
		}
	\end{center}
\end{table}
For example, $ucov({ACD}, T_1)$ in the database in Figure~\ref{fig:Example_2} is 37. \\ $sumcov(\lbrace ACD \rbrace, T_1)$ can be computed as $5 \times Co(\lbrace A \rbrace) + 10 \times Co(\lbrace C \rbrace) + 2 \times Co(\lbrace D \rbrace) = 5\times 4 + 10\times 3 + 2\times 4 = 58$. 

\section{Experiments and results}\label{sec:Experiments and Results}
In this section, we compare the performance of algorithms from the category of tree-based, list-based, and projection-based algorithms. We specifically ask which category algorithms performs best on Sparse and Dense datasets? Do we get the same performance trends as we get for HUIM? The results of different datasets are shown in Figure \ref{fig:real_sparse_perf_coverage}.
 
\noindent \textbf{Experimental Setup:} We choose UP-Growth+ \cite{up_tree} from tree-based, SM-Miner, EFIM \cite{efim} and D2HUP \cite{d2hup} from projection-based algorithms. We obtained the Java source code of UP-Growth+, EFIM and D2HUP algorithm from the SPMF library \cite{spmf} and adapt them to mine patterns for any arbitrarily subadditive monotone function as described in Section 5. We call the algorithms adapted for SM functions as UPG+SM, EFIMSM, and D2HUPSM respectively. We show the experiment results to observe if there is any change in performance trend compared to HUIM for a complex recursive utility function like $ucov$.

The experiments were performed on an Intel Xeon(R) CPU=26500@2.00 GHz with 64 GB RAM and Windows Server 2012 operating system. We conduct experiments on the ChainStore, Kosarak, Mushroom, and Accidents datasets obtained from the SPMF library \cite{spmf}. The datasets vary in the number of transactions, the number of items, and the average transaction length as shown in Table \ref{Tab:datasets}. 
\begin{table}
	
	\begin{center}
		\caption{$ Characteristics \, of \, real \, datasets$ \label{Tab:datasets}}
		\begin{tabular}{ l l  l l  l}
			\hline
			\bfseries{Dataset} & \bfseries{\#Tx} & \bfseries{Avg. length}& \bfseries{\#Items} &  \bfseries{Type}  \\ \hline
			Chainstore & 11,12,949 & 7.2 & 46,086 & Sparse  \\
			Kosarak & 9,90,002 & 8.1 & 41,270 & Sparse \\
			Mushroom & 8,124 & 23 & 119 & Dense \\
			Accidents & 3,40,183 & 33.8 & 468 & Dense \\ \hline
		\end{tabular}
	\end{center}
\end{table}
The coverage graph was constructed from the transaction database by taking the set of distinct items present in the database as vertices and linking them by edges. The average degree of a vertex in the coverage graph is four. The metrics for performance measure are total execution time, the number of explored candidates, and the number of utility function calls (UFC). Every algorithm is executed five times, and an average is taken for the total execution time. In our experiments, the utility values are expressed in terms of percentage of the utility threshold at which at least one candidate itemset is reported.

\noindent \textbf{Result-1: Which category algorithm performs better on Sparse datasets?} We observe that tree-based algorithm UPG+SM and list-based algorithm SM-Miner perform better than projection based algorithms EFIMSM and D2HUPSM on Sparse datasets. The reason for this behavior is less number of utility function calls by UPG+SM. The UPG+SM algorithm does not call the utility function during the candidate generation phase. The function calls are made only during the tree construction and verification phase. However, EFIMSM and D2HUPSM identify promising items during every recursive call, unlike SM-Miner and UPG+SM. Hence, the number of function calls by EFIMSM, D2HUPSM is more compared to SM-Miner and UPG+SM. There is an outlier at 14 $\%$ threshold on Kosarak dataset which is due to a large number of candidates generated by UPG+SM and the overhead for candidate verification dominates the execution time. 

\noindent \textbf{Result-2: Which category algorithm performs better on Dense}\\ \textbf{datasets?} SM-Miner performs the best on the Mushroom dataset and EFIMSM performs the best on Accidents dataset. UPG+SM performs the worst for Mushroom and Accidents dense datasets as it generates lots of candidates, and the verification phase dominates the runtime performance for tree-based algorithms. The execution for UPG+SM didn't complete on Accidents dataset even after 24 hours and was terminated. EFIMSM performs the best for Accidents dataset as it generates the least number of candidates as well as the number of calls to the utility computation function. SM-Miner performs the best for Mushroom dataset followed by EFIMSM.

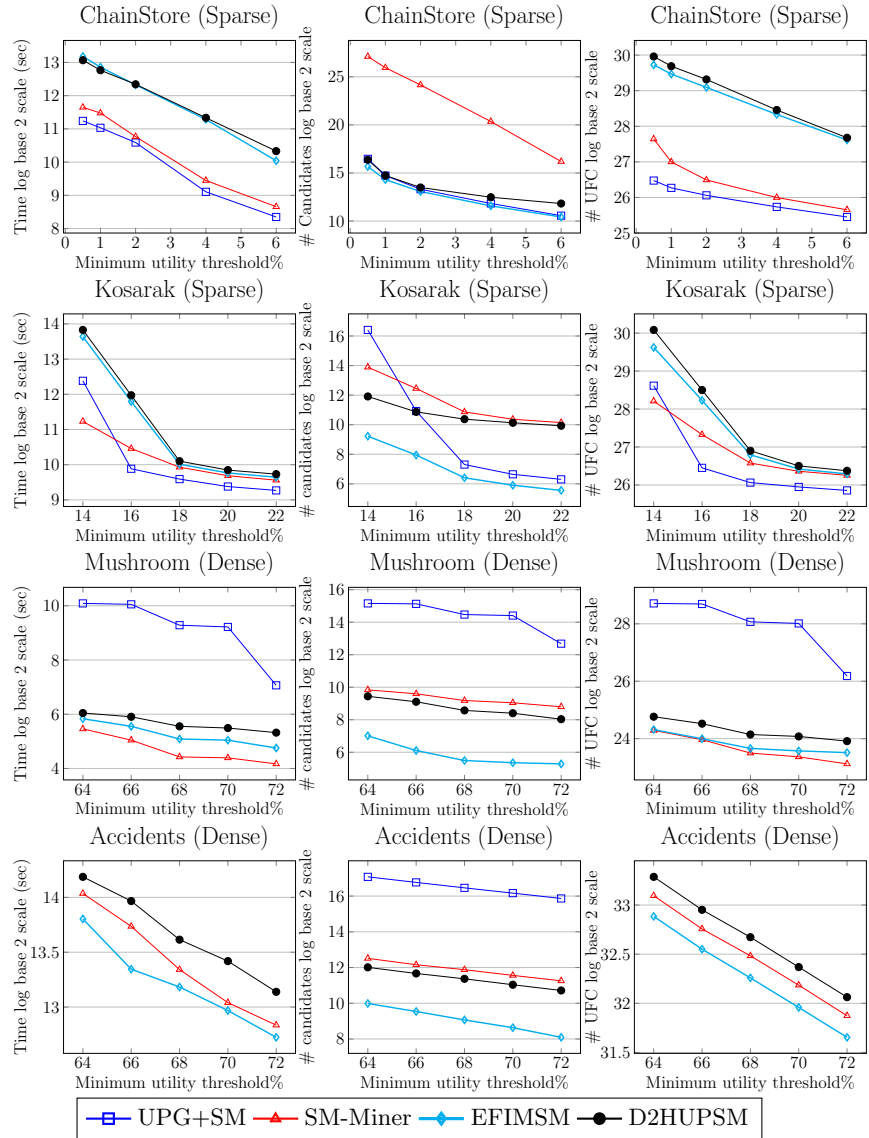
\begin{figure}
	\begin{center}
		\begin{tabular}{ccc}
			\centering
			
			\begin{minipage}{.27\linewidth}
				\begin{tikzpicture}[scale=0.45]
				\begin{axis}[
				title={\huge ChainStore (Sparse)},
				xlabel={Minimum utility threshold$\%$ },
				ylabel={Time log base 2 scale (sec)},
				enlarge y limits=true,
					legend to name=t2,
				legend style={at={($(0,0)+(1cm,1cm)$)},
					legend columns=4,fill=none,draw=black,anchor=center,align=center},
				ymajorgrids=true,
				mark size=3pt,
				label style={font=\Large},
                tick label style={font=\Large} 
				]
				
				\addplot[
				color=blue,
				mark=square,
				style=thick,
				]
				coordinates {
					(0.5,{log2(2419)})(1,{log2(2096)})(2,{log2(1541)})(4,{log2(551)})(6,{log2(326)})
				};
				\addlegendentry{UPG+SM}
				\addplot[
				color=red,
				mark=triangle,
				style=thick,
				]
				coordinates {
					(0.5,{log2(3214)})(1,{log2(2851)})(2,{log2(1743)})(4,{log2(698)})(6,{log2(404)})
				};    
				\addlegendentry{SM-Miner}
				\addplot[
				color=cyan,
				mark=diamond,
				style=very thick,
				]
				coordinates {
					(0.5,{log2(9272)})(1,{log2(7428)})(2,{log2(5136)})(4,{log2(2492)})(6,{log2(1054)})
				};    
				\addlegendentry{EFIMSM}
				\addplot[
				color=black,
				mark=*,
				style=thick,
				]
				coordinates {
					(0.5,{log2(8580)})(1,{log2(6961)})(2,{log2(5194)})(4,{log2(2582)})(6,{log2(1288)})
				};    
				\addlegendentry{D2HUPSM}
				\end{axis}
				
				\end{tikzpicture}
			\end{minipage}
			\hspace{0.3cm}
			
			&
			
			\begin{minipage}{.27\linewidth}
				\begin{tikzpicture}[scale=0.45]
				\begin{axis}[
				title={\huge ChainStore (Sparse)},
				xlabel={Minimum utility threshold$\%$ },
				ylabel={$\#$ Candidates log base 2 scale},
				enlarge y limits=true,
				ymajorgrids=true,
				,mark size=3pt,
				label style={font=\Large},
                tick label style={font=\Large}
				]
				
				\addplot[
				color=blue,
				mark=square,
				style=thick,
				]
				coordinates {
					(0.5,{log2(90783)})(1,{log2(27157)})(2,{log2(10150)})(4,{log2(3641)})(6,{log2(1511)})
				};
				\addplot[
				color=red,
				mark=triangle,
				style=thick,
				]
				coordinates {
					(0.5,{log2(145172035)})(1,{log2(64734132)})(2,{log2(18670528)})(4,{log2(1339816)})(6,{log2(74819)})
				};    
				\addplot[
				color=cyan,
				mark=diamond,
				style=very thick,
				]
				coordinates {
					(0.5,{log2(51893)})(1,{log2(20229)})(2,{log2(8638)})(4,{log2(3064)})(6,{log2(1375)})
				};    
				\addplot[
				color=black,
				mark=*,
				style=thick,
				]
				coordinates {
					(0.5,{log2(85121)})(1,{log2(27305)})(2,{log2(11505)})(4,{log2(5709)})(6,{log2(3625)})
				};
				\end{axis}
				
				\end{tikzpicture}
			\end{minipage}
			\hspace{0.3cm}
			&
			
			\begin{minipage}{.27\linewidth}
				\begin{tikzpicture}[scale=0.45]
				\begin{axis}[
				title={\huge ChainStore (Sparse)},
				xlabel={Minimum utility threshold$\%$ },
				ylabel={$\#$ UFC log base 2 scale},
				enlarge y limits=true,
				ymajorgrids=true,
				,mark size=3pt,
				label style={font=\Large},
                tick label style={font=\Large}
				]
				
				\addplot[
				color=blue,
				mark=square,
				style=thick,
				]
				coordinates {
					(0.5,{log2(93276703)})(1,{log2(80932114)})(2,{log2(69969499)})(4,{log2(55875453)})(6,{log2(45910088)})
				};
				\addplot[
				color=red,
				mark=triangle,
				style=thick,
				]
				coordinates {
					(0.5,{log2(209442672)})(1,{log2(134196734)})(2,{log2(94394372)})(4,{log2(67067830)})(6,{log2(52782291)})
				};    
				\addplot[
				color=cyan,
				mark=diamond,
				style=very thick,
				]
				coordinates {
					(0.5,{log2(886982459)})(1,{log2(742137113)})(2,{log2(572334499)})(4,{log2(338189340)})(6,{log2(205538744)})
				};    
				\addplot[
				color=black,
				mark=*,
				style=thick,
				]
				coordinates {
					(0.5,{log2(1044587602)})(1,{log2(864787265)})(2,{log2(668144909)})(4,{log2(368702200)})(6,{log2(214565320)})
				};
				\end{axis}
				
				\end{tikzpicture}
			\end{minipage}
			\\
			\begin{minipage}{.27\linewidth}
				\begin{tikzpicture}[scale=0.45]
				\begin{axis}[
				title={\huge Kosarak (Sparse)},
				xlabel={Minimum utility threshold$\%$ },
				ylabel={Time log base 2 scale (sec)},
				enlarge y limits=true,
				ymajorgrids=true,
				,mark size=3pt,
				label style={font=\Large},
                tick label style={font=\Large}
				]
				
				\addplot[
				color=blue,
				mark=square,
				style=thick,
				]
				coordinates {
					(14,{log2(5343)})(16,{log2(945)})(18,{log2(771)})(20,{log2(666)})(22,{log2(617)})
				};
				\addplot[
				color=red,
				mark=triangle,
				style=thick,
				]
				coordinates {
					(14,{log2(2401)})(16,{log2(1408)})(18,{log2(975)})(20,{log2(824)})(22,{log2(755)})
				};    
				\addplot[
				color=cyan,
				mark=diamond,
				style=very thick,
				]
				coordinates {
					(14,{log2(12760)})(16,{log2(3553)})(18,{log2(1034)})(20,{log2(870)})(22,{log2(800)})
				};    
				\addplot[
				color=black,
				mark=*,
				style=thick,
				]
				coordinates {
					(14,{log2(14545)})(16,{log2(4018)})(18,{log2(1098)})(20,{log2(921)})(22,{log2(850)})
				};    
				\end{axis}
				
				\end{tikzpicture}
			\end{minipage}
			\hspace{0.3cm}
			&
			
			\begin{minipage}{.27\linewidth}
				\begin{tikzpicture}[scale=0.45]
				\begin{axis}[
				title={\huge Kosarak (Sparse)},
				xlabel={Minimum utility threshold$\%$ },
				ylabel={$\#$ candidates log base 2 scale},
				enlarge y limits=true,
				ymajorgrids=true,
				,mark size=3pt,
				label style={font=\Large},
                tick label style={font=\Large}
				]
				
				\addplot[
				color=blue,
				mark=square,
				style=thick,
				]
				coordinates {
					(14,{log2(87008)})(16,{log2(1942)})(18,{log2(158)})(20,{log2(100)})(22,{log2(79)})
				};
				\addplot[
				color=red,
				mark=triangle,
				style=thick,
				]
				coordinates {
					(14,{log2(15206)})(16,{log2(5599)})(18,{log2(1867)})(20,{log2(1321)})(22,{log2(1131)})
				};    
				\addplot[
				color=cyan,
				mark=diamond,
				style=very thick,
				]
				coordinates {
					(14,{log2(597)})(16,{log2(247)})(18,{log2(85)})(20,{log2(60)})(22,{log2(47)})
				};    
				\addplot[
				color=black,
				mark=*,
				style=thick,
				]
				coordinates {
					(14,{log2(3847)})(16,{log2(1869)})(18,{log2(1327)})(20,{log2(1121)})(22,{log2(975)})
				};
				\end{axis}
				
				\end{tikzpicture}
			\end{minipage}
			\hspace{0.3cm}
			&
			
			\begin{minipage}{.27\linewidth}
				\begin{tikzpicture}[scale=0.45]
				\begin{axis}[
				title={\huge Kosarak (Sparse)},
				xlabel={Minimum utility threshold$\%$ },
				ylabel={$\#$ UFC log base 2 scale},
				enlarge y limits=true,
				ymajorgrids=true,
				mark size=3pt,
				label style={font=\Large},
                tick label style={font=\Large}
				]
				
				\addplot[
				color=blue,
				mark=square,
				style=thick,
				]
				coordinates {
					(14,{log2(411223356)})(16,{log2(91756496)})(18,{log2(70094850)})(20,{log2(64780576)})(22,{log2(60763956)})
				};
				\addplot[
				color=red,
				mark=triangle,
				style=thick,
				]
				coordinates {
					(14,{log2(310140616)})(16,{log2(168402729)})(18,{log2(100077884)})(20,{log2(86192874)})(22,{log2(80038778)})
				};    
				\addplot[
				color=cyan,
				mark=diamond,
				style=very thick,
				]
				coordinates {
					(14,{log2(826427368)})(16,{log2(314247289)})(18,{log2(116382358)})(20,{log2(89755606)})(22,{log2(82147382)})
				};    
				\addplot[
				color=black,
				mark=*,
				style=thick,
				]
				coordinates {
					(14,{log2(1137140909)})(16,{log2(379781409)})(18,{log2(125325787)})(20,{log2(95070330)})(22,{log2(86964310)})
				};
				\end{axis}
				
				\end{tikzpicture}
			\end{minipage}\\ 
			\begin{minipage}{.27\linewidth}
				\begin{tikzpicture}[scale=0.45]
				\begin{axis}[
				title={\huge Mushroom (Dense)},
				xlabel={Minimum utility threshold$\%$ },
				ylabel={Time log base 2 scale (sec)},
				enlarge y limits=true,
				ymajorgrids=true,
				mark size=3pt,
				label style={font=\Large},
                tick label style={font=\Large}
				]
				
				\addplot[
				color=blue,
				mark=square,
				style=thick,
				]
				coordinates {
					(64,{(log2(1088))})(66,{(log2(1062))})(68,{(log2(623))})(70,{(log2(596))})(72,{(log2(134))})
				};
				\addplot[
				color=red,
				mark=triangle,
				style=thick,
				]
				coordinates {
					(64,{(log2(44))})(66,{(log2(33))})(68,{(log2(21.5))})(70,{(log2(21))})(72,{(log2(18))})
				};    
				\addplot[
				color=cyan,
				mark=diamond,
				style=very thick,
				]
				coordinates {
					(64,{(log2(57))})(66,{(log2(47))})(68,{(log2(34))})(70,{(log2(33))})(72,{(log2(27))})
				};    
				\addplot[
				color=black,
				mark=*,
				style=thick,
				]
				coordinates {
					(64,{log2(66)})(66,{log2(60)})(68,{log2(47)})(70,{log2(45)})(72,{log2(40)})
				};    
				\end{axis}
				
				\end{tikzpicture}
			\end{minipage}
			\hspace{0.3cm}
			&
			
			\begin{minipage}{.27\linewidth}
				\begin{tikzpicture}[scale=0.45]
				\begin{axis}[
				title={\huge Mushroom (Dense)},
				xlabel={Minimum utility threshold$\%$ },
				ylabel={$\#$ candidates log base 2 scale},
				enlarge y limits=true,
				ymajorgrids=true,
				mark size=3pt,
				label style={font=\Large},
                tick label style={font=\Large}
				]
				
				\addplot[
				color=blue,
				mark=square,
				style=thick,
				]
				coordinates {
					(64,{(log2(36524))})(66,{(log2(35846))})(68,{(log2(22696))})(70,{(log2(21724))})(72,{(log2(6556))})
				};
				\addplot[
				color=red,
				mark=triangle,
				style=thick,
				]
				coordinates {
					(64,{(log2(915))})(66,{(log2(774))})(68,{(log2(581))})(70,{(log2(528))})(72,{(log2(445))})
				};    
				\addplot[
				color=cyan,
				mark=diamond,
				style=very thick,
				]
				coordinates {
					(64,{(log2(129))})(66,{(log2(69))})(68,{(log2(45))})(70,{(log2(41))})(72,{(log2(39))})
				};    
				\addplot[
				color=black,
				mark=*,
				style=thick,
				]
				coordinates {
					(64,{log2(695)})(66,{log2(551)})(68,{log2(381)})(70,{log2(338)})(72,{log2(262)})
				};
				
				\end{axis}
				
				\end{tikzpicture}
			\end{minipage}
			\hspace{0.3cm}
			&
			
			\begin{minipage}{.27\linewidth}
				\begin{tikzpicture}[scale=0.45]
				\begin{axis}[
				title={\huge Mushroom (Dense)},
				xlabel={Minimum utility threshold$\%$ },
				ylabel={$\#$ UFC log base 2 scale},
				enlarge y limits=true,
				ymajorgrids=true,
				mark size=3pt,
				label style={font=\Large},
                tick label style={font=\Large}
				]
				
				\addplot[
				color=blue,
				mark=square,
				style=thick,
				]
				coordinates {
					(64,{(log2(439612943))})(66,{(log2(433293671))})(68,{(log2(281661954))})(70,{(log2(270944837))})(72,{(log2(76154252))})
				};
				\addplot[
				color=red,
				mark=triangle,
				style=thick,
				]
				coordinates {
					(64,{(log2(20467267))})(66,{(log2(16356086))})(68,{(log2(11849276))})(70,{(log2(10811239))})(72,{(log2(9122794))})
				};    
				\addplot[
				color=cyan,
				mark=diamond,
				style=very thick,
				]
				coordinates {
					(64,{(log2(20859324))})(66,{(log2(16740151))})(68,{(log2(13253905))})(70,{(log2(12476142))})(72,{(log2(11970321))})
				};    
				\addplot[
				color=black,
				mark=*,
				style=thick,
				]
				coordinates {
					(64,{log2(28501275)})(66,{log2(24086726)})(68,{log2(18535253)})(70,{log2(17697158)})(72,{log2(15800471)})
				};
				
				\end{axis}
				
				\end{tikzpicture}
			\end{minipage}
			\\
			\begin{minipage}{.27\linewidth}
				\begin{tikzpicture}[scale=0.45]
				\begin{axis}[
				title={\huge Accidents (Dense)},
				xlabel={Minimum utility threshold$\%$ },
				ylabel={Time log base 2 scale (sec)},
				enlarge y limits=true,
				ymajorgrids=true,
				mark size=3pt,
				label style={font=\Large},
                tick label style={font=\Large}
				]
				
				\addplot[
				color=red,
				mark=triangle,
				style=thick,
				]
				coordinates {
					(64,{(log2(16795))})(66,{(log2(13643))})(68,{(log2(10383))})(70,{(log2(8425))})(72,{(log2(7309))})
				};    
				\addplot[
				color=cyan,
				mark=diamond,
				style=very thick,
				]
				coordinates {
					(64,{(log2(14300))})(66,{(log2(10407))})(68,{(log2(9299))})(70,{(log2(8013))})(72,{(log2(6766))})
				};    
				\addplot[
				color=black,
				mark=*,
				style=thick,
				]
				coordinates {
					(64,{log2(18658)})(66,{log2(16009)})(68,{log2(12542)})(70,{log2(10952)})(72,{log2(9020)})
				};
				
				\end{axis}
				
				\end{tikzpicture}
			\end{minipage}
			\hspace{0.3cm}
			&
			
			\begin{minipage}{.27\linewidth}
				\begin{tikzpicture}[scale=0.45]
				\begin{axis}[
				title={\huge Accidents (Dense)},
				xlabel={Minimum utility threshold$\%$ },
				ylabel={$\#$ candidates log base 2 scale},
				enlarge y limits=true,
				ymajorgrids=true,
				mark size=3pt,
				label style={font=\Large},
                tick label style={font=\Large}
				]
				
				\addplot[
				color=blue,
				mark=square,
				style=thick,
				]
				coordinates {
					(64,{(log2(138281))})(66,{(log2(111489))})(68,{(log2(90228))})(70,{(log2(73515))})(72,{(log2(59779))})
				};
				\addplot[
				color=red,
				mark=triangle,
				style=thick,
				]
				coordinates {
					(64,{(log2(5848))})(66,{(log2(4568))})(68,{(log2(3767))})(70,{(log2(3021))})(72,{(log2(2445))})
				};    
				\addplot[
				color=cyan,
				mark=diamond,
				style=very thick,
				]
				coordinates {
					(64,{(log2(1015))})(66,{(log2(744))})(68,{(log2(536))})(70,{(log2(397))})(72,{(log2(273))})
				};    
				\addplot[
				color=black,
				mark=*,
				style=thick,
				]
				coordinates {
					(64,{log2(4125)})(66,{log2(3254)})(68,{log2(2644)})(70,{log2(2102)})(72,{log2(1681)})
				};
				
				\end{axis}
				
				\end{tikzpicture}
			\end{minipage}
			\hspace{0.3cm}
			&
			
			\begin{minipage}{.27\linewidth}
				\begin{tikzpicture}[scale=0.45]
				\begin{axis}[
				title={\huge Accidents (Dense)},
				xlabel={Minimum utility threshold$\%$ },
				ylabel={$\#$ UFC log base 2 scale},
				enlarge y limits=true,
				ymajorgrids=true,
				mark size=3pt,
				label style={font=\Large},
                tick label style={font=\Large}
				]
				
				\addplot[
				color=red,
				mark=triangle,
				style=thick,
				]
				coordinates {
					(64,{(log2(9175572791))})(66,{(log2(7258629338))})(68,{(log2(6004441884))})(70,{(log2(4881941798))})(72,{(log2(3935664480))})
				};    
				\addplot[
				color=cyan,
				mark=diamond,
				style=very thick,
				]
				coordinates {
					(64,{(log2(7929969344))})(66,{(log2(6292722696))})(68,{(log2(5139381998))})(70,{(log2(4172052247))})(72,{(log2(3379561481))})
				};    
				\addplot[
				color=black,
				mark=*,
				style=thick,
				]
				coordinates {
					(64,{log2(10480176493)})(66,{log2(8306843707)})(68,{log2(6848019619)})(70,{log2(5545724096)})(72,{log2(4486222798)})
				};
				
				\end{axis}
				
				\end{tikzpicture}
			\end{minipage}\\


			
		\end{tabular}
		
		\ref{t2}
		\caption{Performance evaluation on real datasets for $ucov$.}
		\label{fig:real_sparse_perf_coverage}
	\end{center}
\end{figure}

\noindent \textbf{Result-3: Are the performance trend of algorithms same as HUIM scenario?} The performance trend of algorithms is different from HUIM scenario. Projection-based algorithms are known to perform an order of magnitude better than other algorithms on dense and sparse datasets for HUIM. D2HUP performs the best on Sparse and EFIM performs the best on dense datasets for HUIM. However, we observe that the tree-based and list-based algorithms perform better compared to projection-based algorithms on sparse datasets. SM-Miner competes with EFIMSM for the best performance on dense datasets. We observe that the total execution time of the algorithms is more correlated with the number of utility function calls compared to the number of candidates generated during the mining process. 

\section{Conclusion}\label{sec:Conclusion and Future Work}
In this paper, we formalize the problem of high-utility itemset mining for a general utility function. We focus on designing upper-bounds for high-utility itemset mining problem with subadditive monotone functions (HUIM-SM). We propose a new inverted-list data structure called SMI-List with a lightweight construction method and an algorithm called SM-Miner to find high-utility itemsets for arbitrary subadditive monotone functions. We also adapt the existing tree-based and projection-based high-utility itemset mining algorithms for SM functions. We compare the performance of several HUIM-SM algorithms on dense and sparse datasets. Our results demonstrate that the computation of utility can play an important role in the runtime performance of algorithms for complex utility functions like $ucov$. Subadditive monotone utility functions like $ucov$ can be designed to find high-influence active groups that can be attractive for applications like viral marketing.

\bibliographystyle{unsrt}
\bibliography{bib_file_backup}
\end{document}